\documentclass{cta-author}

\usepackage{graphicx}
\usepackage{amsmath}
\usepackage{amssymb}
\usepackage{multirow}
\usepackage{longtable}
\usepackage{subfigure}
\usepackage{float}
\usepackage{afterpage}
\usepackage{algorithm}
\usepackage{subeqnarray}
\usepackage{algpseudocode}
\usepackage{epsfig}
\usepackage{epstopdf}
\usepackage{natbib}

\newtheorem{theorem}{Theorem}{}
{}
{}
\newtheorem{proposition}{Proposition}{}
\newtheorem{lemma}{Lemma}{}

\begin{document}


\title{Weighted Sum-Throughput Maximization for Energy Harvesting Powered MIMO Multi-Access Channels}

\author{\au{Zheng Nan$^{1\corr}$}, \au{Wenming Li$^{2}$}}

\address{\add{1}{RF Tech. R\&D Laboratory, AVIC Beijing Keveen Aviation Instrument Co., Ltd., 43 North 3rd Ring Road West, Beijing, 100086, P.R.China}
\add{2}{SPD Bank, 1688 Lianhua Road, Shanghai, 200233, P.R.China}
\email{znan12@fudan.edu.cn}}

\begin{abstract}
This paper develops a novel approach to obtaining the optimal scheduling strategy in a multi-input multi-output
(MIMO) multi-access channel (MAC), where each transmitter is powered by an individual energy harvesting process. Relying on the state-of-the-art convex optimization tools, the proposed approach provides a low-complexity block coordinate ascent algorithm to obtain the optimal transmission policy that maximizes the weighted sum-throughput for MIMO MAC. The proposed approach can provide the optimal benchmarks for all practical schemes in energy-harvesting powered MIMO MAC transmissions. Based on the revealed structure of the optimal policy, we also propose an efficient online scheme, which requires only causal knowledge of energy arrival realizations. Numerical results are provided to demonstrate the merits of the proposed novel scheme.
\end{abstract}

\maketitle

\section{Introduction}\label{sec1}

Energy harvesting is an effective solution for prolonging the operating lifetime of self-sustainable wireless networks, and it has attracted growing research interest in recent years \cite{Lei09, Gatzianas10}. The wireless terminals with embedded energy harvesting devices and rechargeable batteries are able to harvest renewable energy from environmental sources such as solar and wind \cite{Kansal07, ChenMag15}. The emergence of a new technique, known as simultaneous wireless information and power transfer (SWIPT), even makes it possible for wireless terminals to harvest energy from the ambient radio-frequency (RF) signals \cite{Fang15, Feng15}.
Different from traditional communication systems, due to the intermittent nature of most renewable energy sources, an energy availability constraint is imposed such that the energy accumulatively consumed up to any time cannot exceed what has been accumulatively harvested so far. With this new type of constraints taken into consideration, the optimal transmission polices for an energy harvesting node in time-invariant point-to-point channels were derived in \cite{Yan12, Ho12, Sharma10} without battery-capacity constraint, and in \cite{Tut12} with a finite battery capacity constraint.
A directional water-filling approach was developed to obtain the optimal packet transmission strategy over time-varying fading channels in \cite{Oze11}, while a dynamic string tautening algorithm was proposed to generate the most energy-efficient schedule for delay-limited traffic of transmitters with non-negligible circuit power in \cite{Chen16}. Based on the optimal transmission policies, efficient online transmission schemes were also developed in \cite{Chen16, Oze11}.
A unified approach to obtaining the optimal transmission schedules for both time-invariant and time-varying fading channels was put forth in \cite{SECON14}.
Generalizing the approaches for point-to-point channels, optimal transmission policies for the energy harvesting powered multi-input multi-output (MIMO) broadcast channels were addressed in \cite{Ant11, Yang12, Oze12, Wang14}.


Different from the point-to-point or broadcasting transmission with a single energy harvesting powered transmitter, there are multiple transmitters powered by individual energy harvesting processes for a multi-access channel (MAC).
These multiple energy harvesting processes can have coupled effect on the users' sum-throughput as well as the transmission strategies; hence, optimization approaches for the cases with single energy harvesting node cannot apply any more. As a result, design and analysis of the optimal scheduling policy for energy-harvesting powered MACs is challenging, and the existing research on this critical issue is in a rather primitive stage. Among the limited number of works concerning the scheduling of energy harvesting powered MACs, \cite{Tutuncuoglu12, Ya12} simply investigated the optimal transmission polices for the {\it two-user single-antenna} MACs.


In this paper, we explore the optimal scheduling for the general energy harvesting powered multi-antenna (i.e., MIMO) MACs. Assuming that full harvested energy and channel information is available, we obtain the optimal (offline) transmission policy that maximizes the sum-throughput of multiple users. Relying on a ``nested optimization'' method \cite{Oze12, Xu13}, we show that the optimal MIMO MAC scheduling problem can be converted into a convex power allocation problem. An iterative block coordinate ascent algorithm is then developed to find the optimal power allocation. Specifically, to bypass the coupling effect resulting from the multiple energy harvesting processes, we compute the optimal power allocation for one user with all other users' powers being fixed per iteration. The problem corresponding to each iteration is similar to finding the optimal power scheduling for an equivalent ``point-to-point'' link between the selected user and the access point, for which the ``string-tautening'' algorithm in \cite{Wang14} can be applied to obtain the solution with a linear computational complexity. By ascending the sum-throughput through optimizing each user's power allocation in a sequential way, the proposed approach is guaranteed to converge to the globally optimal power allocation solution, and consequently the globally optimal MIMO MAC strategy. The proposed approach can provide the optimal benchmarks for practical transmission schemes over energy harvesting powered MIMO MACs.
The revealed structures of the optimal policy are also used to develop an online scheduling scheme which requires only causal knowledge of the harvested energy realizations.

The rest of the paper is organized as follows. Section II describes the system models. Section III presents the proposed novel approach to optimal scheduling for the energy harvesting powered MIMO MAC. Section IV evaluates the proposed scheme with numerical examples. We conclude the paper in Section V.

\section{Notations}
\subsection{Indices, Numbers, and Sets}
\begin{tabular}{p{0.12\columnwidth} p{0.76\columnwidth}}
$K$, $k$ & Number and index of transmitters (users). \\
$N$ & Number of epochs. \\
$i$, $n$ & Indices of epochs. \\
$N_t$ & Number of transmit antennas of each user. \\
$N_r$ & Number of receive antennas of the access point. \\
$q$ & Iteration index of the block coordinate ascent algorithm. \\
$\mathcal{C}_{\text{MAC}}$ & Capacity region of MAC. \\
\end{tabular}

\subsection{Constants}
\begin{tabular}{p{0.12\columnwidth} p{0.76\columnwidth}}
$T$ & Instant that transmission terminates. \\
$\boldsymbol{I}$ & Identity covariance matrix. \\
$E_{\max, k}$ & Battery capacity of user $k$. \\
$\lambda_e$ & Energy arrival rate. \\
$\bar{E}_k$ & Mean of energy amount in each arrival. \\
\end{tabular}

\subsection{Basic Variables}
\begin{tabular}{p{0.12\columnwidth} p{0.76\columnwidth}}
$\boldsymbol{H}_k$ & Channel coefficient matrix from the $k$th user to the access point. \\
$\mathcal{H}$ & $\mathcal{H} := \left [ \boldsymbol{H}_1, \ldots, \boldsymbol{H}_K \right]$. \\
$\boldsymbol{x}_k(t)$ & Transmitted vector signal of user $k$ at time $t$. \\
$\boldsymbol{y}(t)$ & Received complex-baseband signal at the access point. \\
$\boldsymbol{z}(t)$ & Zero mean additive complex-Gaussian. \\
$E_{i, k}$ & Amount of the $i$th energy arrival of user $k$. \\
$\boldsymbol{\mathcal{E}}_k$ & Vector collecting $\{ E_{i, k} \}_{i = 1}^N$. \\
$\boldsymbol{\mathcal{E}}$ & Matrix collecting $\{ \boldsymbol{\mathcal{E}}_k \}_{k = 1}^K$. \\
$L_i$ & Length of the $i$th epoch. \\
$\boldsymbol{\mathcal{L}}$ & Vector collecting $\{ L_i \}_{i = 1}^N$. \\
$E_{n, k}^a$ & Amount of energy that has been harvested so far by user $k$. \\
$E_{n, k}^c$ & Least amount of energy that must be consumed so far by user $k$. \\
$w_k$ & Priority weight of user $k$. \\
$\boldsymbol{w}$ & Weight vector collecting $\{ w_k \}_{k = 1}^K$. \\
$\pi$ & Permutation of user indices $\{1, \ldots, K\}$ such that $w_{\pi(1)} \geq w_{\pi(2)} \geq \cdots \geq w_{\pi(K)}$. \\
\end{tabular}

\subsection{Decision Variables}
\begin{tabular}{p{0.22\columnwidth} p{0.66\columnwidth}}
$\boldsymbol{Q}_{i, k}$ & Transmit covariance matrix of user $k$ within epoch $i$. \\
$\boldsymbol{Q}_i$ & $\boldsymbol{Q}_i := [ \boldsymbol{Q}_{i, 1}, \ldots, \boldsymbol{Q}_{i, K}]$. \\
$\boldsymbol{Q}$ & $\boldsymbol{Q} := [ \boldsymbol{Q}_{1}, \ldots, \boldsymbol{Q}_{N}]$. \\
$P_{i, k}$ & Transmit power of user $k$ within epoch $i$. \\
$\boldsymbol{P}_i$ & Vector collecting $\{ P_{i, k} \}_{k = 1}^K$. \\
$\boldsymbol{P}$ &  Matrix collecting $\{ \boldsymbol{P}_i \}_{i = 1}^N$. \\
$r_k^M(\boldsymbol{Q}_i)$ & Achieved rate of user $k$ within epoch $i$. \\
$\boldsymbol{r}^M(\boldsymbol{Q}_i)$ & Vector collecting $\{ r_k^M(\boldsymbol{Q}_i) \}_{k = 1}^K$. \\
$\boldsymbol{\Lambda}$ & Vector collecting all Lagrange multipliers. \\
$\lambda_{n, k}$, $\mu_{n, k}$ & Lagrange multipliers. \\
$\theta_{i, k}$ & $\theta_{i,k} := \sum_{n=i}^N \lambda_{n,k}^* - \sum_{n=i}^{N} \mu_{n,k}^*$. \\
$\boldsymbol{P}_{i, -k}$ & Vector collecting all power values in $\boldsymbol{P}_i$ except for $P_{i, k}$. \\
$P_{i, k}^{(q)}$ & Optimal value of $P_{i, k}$ in the $q$th iteration. \\
\end{tabular}

\noindent
\begin{tabular}{p{0.22\columnwidth} p{0.66\columnwidth}}
$\lambda_{n, k}^{(q)}$, $\mu_{n, k}^{(q)}$ & Optimal values of $\lambda_{n, k}$ and $\mu_{n, k}$ in the $q$th iteration. \\
$\omega_{i, k}^{(q)}$ & Water-level of user $k$ within epoch $i$ in the $q$th iteration. \\
$\omega_{i, k}^{+(q)}$, $\omega_{i, k}^{-(q)}$ & Constant water-levels to make the $n$th causality and non-overflow constraints of user $k$ become tight at $t_n$ in the $q$th iteration. \\
$t_{\tau}$ & First water-level changing time in Algorithm~1. \\
$\omega^+$, $\omega^-$ & Candidate water-levels in Algorithm~1. \\
$\tau^+$, $\tau^-$ & Indices of energy causality and non-overflow constraints corresponding to candidate water-levels in Algorithm~1, respectively. \\
$W^{(q)}$ & Optimal sum-throughput in the $q$th iteration. \\
$\tilde{P}_{i, k}^*$, $P_{i, k}^{(0)}$ & Initial transmit power of user $k$ within epoch $i$. \\
$\tilde{\boldsymbol{P}}_{i}^*$ & Vector collecting $\{ \tilde{P}_{i, k}^* \}_{k = 1}^K$. \\
$P_{n, k}^+$, $P_{n, k}^-$ & Constant powers to make the $n$th causality and non-overflow constraints of user $k$ become tight at $t_n$. \\
\end{tabular}

\subsection{Functions}
\begin{tabular}{p{0.22\columnwidth} p{0.66\columnwidth}}
$R(\boldsymbol{P}_i)$ & Sum-rate within epoch $i$. \\
$\mathcal{L}(\boldsymbol{P}, \boldsymbol{\Lambda})$ & Lagrangian function. \\
$g(P_{i, k})$ & Partial derivative of $R(\boldsymbol{P}_i)$ with respect to $P_{i, k}$. \\
$A_k(t)$ & Energy arrival curve. \\
$D_{\min, k}(t)$ & Minimum energy departure curve. \\
$D_{k}(t)$ & Energy departure curve. \\
\end{tabular}

\section{Modeling Preliminaries}\label{sec2}

In this section, we outline the MAC and energy harvesting models under consideration.

\subsection{Multi-Access Channel}

\begin{figure}[h]\label{fig_mac}
\centering
\includegraphics[width=0.5\textwidth]{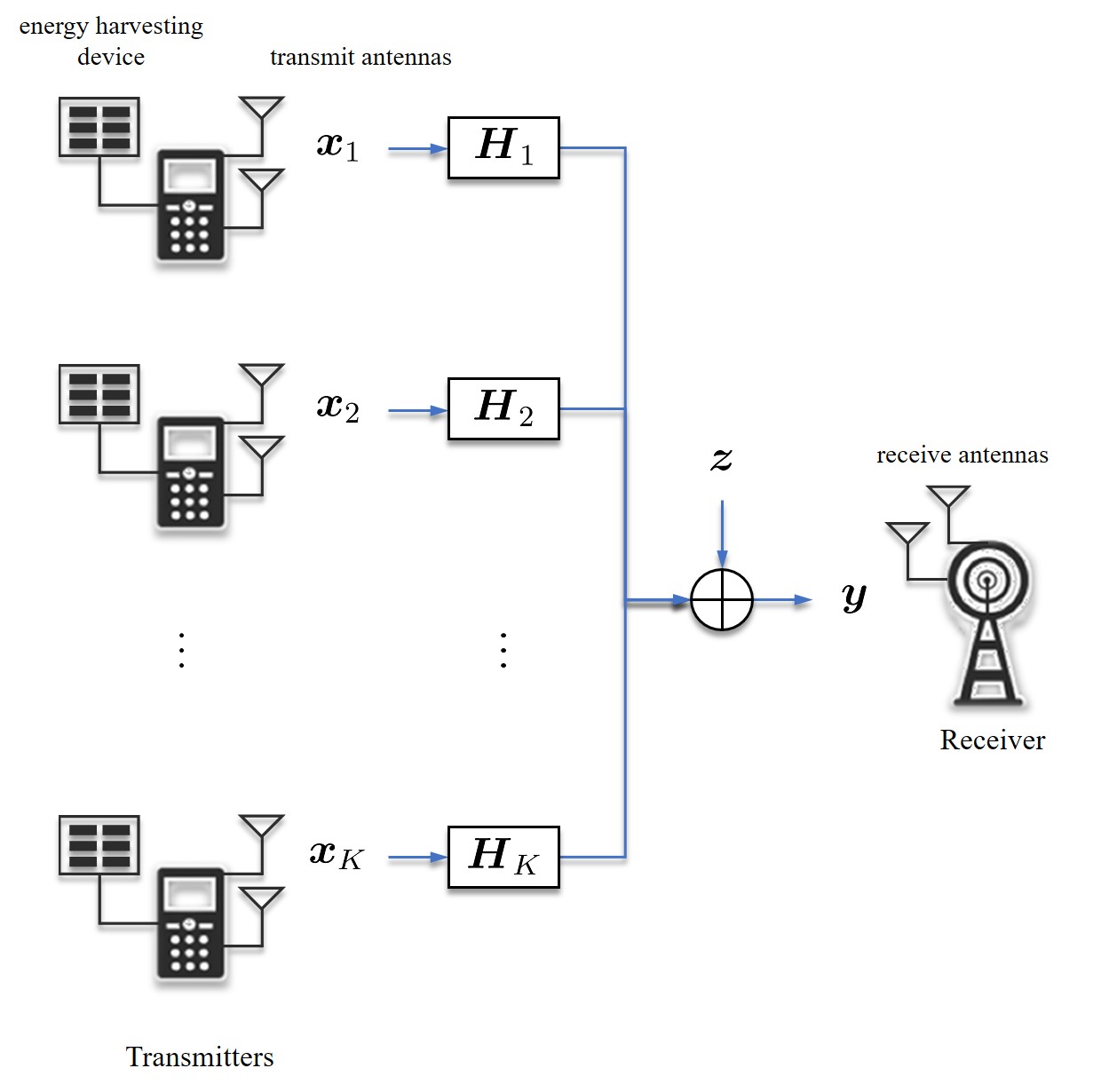}
\caption{An energy harvesting powered multi-access channel.}
\end{figure}

Consider a general MIMO MAC, where an access point serves a total of $K$ transmitters (i.e. users); see Fig.~1. Each of the $K$ users has $N_t$ transmit antennas and the access point has $N_r$ receive antennas. Let $\boldsymbol{H}_k \in \mathbb{C}^{N_r \times N_t}$ denote the channel coefficient matrix from the $k$th user to the access point, $k =1, \ldots, K$. The received complex-baseband signal at the access point is given by:
\begin{equation}
   \boldsymbol{y} (t) = \sum_{k=1}^K \boldsymbol{H}_k \boldsymbol{x}_k(t) + \boldsymbol{z}(t),
\end{equation}
where $\boldsymbol{x}_k(t)$ is the transmitted vector signal of user $k$ at time $t$, and $\boldsymbol{z}(t)$ is additive complex-Gaussian with zero mean and identity covariance matrix $\boldsymbol{I}$ of size $N_r$. Let $\boldsymbol{Q}_k:=\mathbb{E}[\boldsymbol{x}_k\boldsymbol{x}_k^{\dagger}] \succeq 0$ denote the transmit covariance matrix for user k, and let $\boldsymbol{P}:=[P_1, \ldots, P_K]$ collect the transmit-power budgets for all the $K$ users. Let $\mathcal{H} := \left [ \boldsymbol{H}_1, \ldots, \boldsymbol{H}_K \right]$. For a given $\boldsymbol{P}$, the MAC capacity region is:
\begin{equation}\nonumber
\begin{aligned}
& \mathcal{C}_{\text{MAC}}(\boldsymbol{P};\mathcal{H}) = \left . \bigcup_{\{\boldsymbol{Q}_k: \; \text{tr}(\boldsymbol{Q}_{k}) \leq P_{k}, \; \forall k\}} \right \{(r_1, \ldots, r_K): \\
&~~~ \left . \sum_{k \in \mathcal{S}} r_k \leq \log \left|\boldsymbol{I} + \sum_{k \in \mathcal{S}} \boldsymbol{H}_k \boldsymbol{Q}_k \boldsymbol{H}_k^{\dagger}\right|, \; \forall \mathcal{S} \subseteq \{1,\ldots, K\} \right \},
\end{aligned}
\end{equation}
where $r_k$ is the achievable transmission rate of user $k$.

\subsection{Energy Harvesting Process}

Suppose that all users do not have persistent power supply. Instead, with the energy harvesting devices and rechargeable batteries, each of the $K$ users could harvest renewable energy (e.g. solar, wind or RF) from the surrounding environment and then store the energy in the battery for future use, as shown in Fig.~1. Each user is powered by an individual energy harvesting process. 
To be specific, for user $k$, the battery capacity is $E_{\max,k}$, and the initial energy available in the battery (at time $t_0 =0$) is denoted by $E_{0,k}$. Over the entire transmission interval $[0, T]$, suppose that there are $N-1$ energy arrivals occurring at time $\{t_1, t_2, \ldots, t_{N-1}\}$ in amounts $\{E_{1,k}, E_{2,k}, \ldots, E_{N-1,k}\}$. Note that we only require $E_{i,k}>0$ for a certain $k$ at any $t_i$; i.e., it is allowed that $E_{i,k}=0$ without loss of generality (w.l.o.g.) for some $k$, implying the independence of the $K$ users' energy harvesting processes; see an illustration in Fig.~2. For convenience, let $t_N = T$. The time interval between two consecutive energy arrivals is defined as an epoch.
The length of the $i$th epoch is then $L_i = t_i- t_{i-1}$, $i = 1, \ldots, N$. It is clear that we have $0 < E_{i,k} \leq E_{\max,k}$, $i = 0, 1, \ldots, N-1$, $k = 1, \ldots, K$; otherwise, the excess energy $E_{i,k} - E_{\max,k}$ cannot be stored in the battery and w.l.o.g. we can set $E_{i,k}=E_{\max,k}$ in such cases. 


\begin{figure}[h]\label{fig_mac}
\centering
\includegraphics[width=0.42\textwidth]{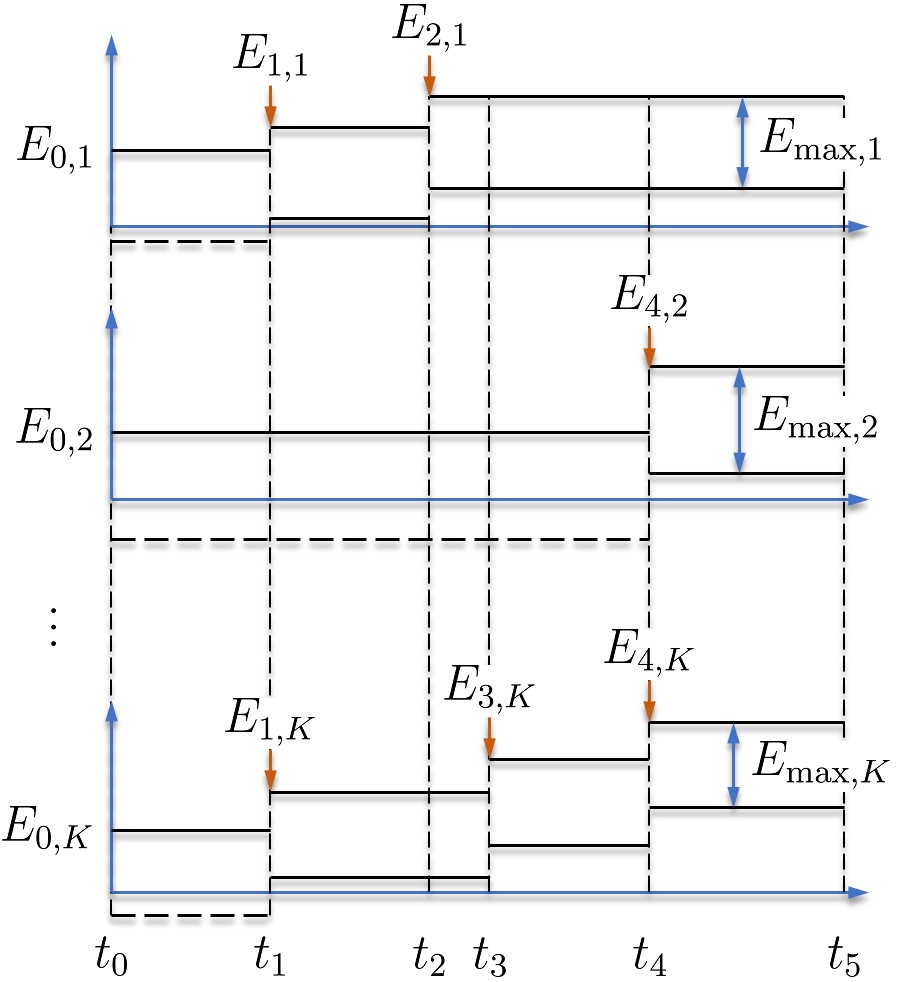}
\caption{Energy harvesting processes of the $K$ users.}
\end{figure}

\section{Optimal Scheduling for Energy-Harvesting Powered MIMO MAC}\label{sec3}

Consider a time-invariant channel $\mathcal{H}:=\{ \boldsymbol{H}_1, \ldots, \boldsymbol{H}_K \}$ over transmission interval $[0, T]$. Let $\boldsymbol{Q}_i:=[\boldsymbol{Q}_{i,1}, \ldots, \boldsymbol{Q}_{i,K}]$ collect the transmit-covariance matrices during the $i$th epoch and $\boldsymbol{Q}:=[\boldsymbol{Q}_1, \ldots, \boldsymbol{Q}_N]$. With the covariance matrix $\boldsymbol{Q}_{i,k}$, the transmit-power for user $k$ within epoch $i$ is then $P_{i,k} := \text{tr}(\boldsymbol{Q}_{i,k})$, where $\text{tr}(\cdot)$ denotes the trace operation. Let $r_k^M(\boldsymbol{Q}_i)$ denote the achieved rate for user $k$ within epoch $i$, and $\boldsymbol{r}^M(\boldsymbol{Q}_i) := [r_1^M(\boldsymbol{Q}_i),\ldots,r_K^M(\boldsymbol{Q}_i)]$. For convenience, we define $E_{n,k}^a:=\sum_{i=0}^{n-1}E_{i,k}$, and $ E_{n,k}^c:=(\sum_{i=0}^{n}E_{i,k} - E_{\max,k})^+,\forall k,\forall n$, where $E_{N, k} := E_{\max, k}, \forall k$. Provided a priority weight vector $\boldsymbol{w} :=[w_1,\ldots,w_K]$, we aim to maximize the weighted-sum of user throughput:
\begin{equation}\label{eq.p}
\begin{array}{cl}
\displaystyle \max_{\boldsymbol{Q}} & \displaystyle \sum_{k=1}^{K}{\left [w_k \sum_{i=1}^{N}{(r_k^M(\boldsymbol{Q}_i) L_i)} \right ]} \\
\displaystyle \text{s. t.} & \displaystyle \text{(C1): } \sum_{i=1}^{n}(P_{i, k} L_i) \leq E_{n, k}^{a}, \;\; \forall n, \;\; \forall k, \\
\displaystyle & \displaystyle \text{(C2): } \sum_{i=1}^{n}(P_{i, k} L_i) \geq E_{n, k}^{c}, \;\; \forall n, \;\; \forall k, \\
\displaystyle & \displaystyle \text{(C3): } \boldsymbol{r}^M(\boldsymbol{Q}_i) \in \mathcal{C}_{\text{MAC}} (\boldsymbol{P}_i; \mathcal{H}), \;\; \forall n, \;\; \forall k. \\
\end{array}
\end{equation}
Here, the first set of constraints (C1) are the energy causality constraints: the total amount of energy consumed by user $k$ up to any arrival time $t_n$ cannot be greater than $E_{n, k}^a$, which is what has been accumulatively harvested so far by user $k$. The second set of constraints (C2) are the non-overflow constraints: the total amount of energy consumed by user $k$ up to $t_n$ cannot be less than $E_{n, k}^c$, in order to avoid the waste of energy \cite{Tut12}.

Due to the finite battery capacity, an energy overflow occurs if the sum of unconsumed energy and newly arriving energy with any user $k$ exceeds $E_{\max, k}$ at the time of energy arrival. Since any transmission policy causing energy overflows can be dominated by a policy without such overflows, the optimal power allocation must satisfy the non-overflow constraints. Hence, the optimal policy can be found among the policies that satisfy all the energy causality and non-overflow constraints in \eqref{eq.p}. Note that with $E_{N, k} = E_{\max, k}, \forall k$, the non-overflow and causality constraints at $t_N$ render $\sum_{i = 1}^N (P_{i, k} L_i) = \sum_{i = 0}^{N - 1} E_{i, k}, \forall k$, i.e., all the energy harvested must be used up at the end per user $k$.

With $R(\boldsymbol{P}_i) := \max_{\boldsymbol{r}^M(\boldsymbol{Q}_i) \in \mathcal{C}_{\text{MAC}}(\boldsymbol{P}_i; \mathcal{H})} \; \sum_{k=1}^K w_k r_k^M(\boldsymbol{Q}_i)$, we can establish that:
\begin{lemma}\label{lemma1}
{\it The strictly concave function $R(\boldsymbol{P}_i)$ can be alternatively obtained as the optimal value of the following convex problem:
\begin{equation}\label{eq.Rp}
\begin{array}{cl}
\displaystyle \max_{\boldsymbol{Q}_{i, k} \succeq 0} & \displaystyle \left . \sum_{k=1}^{K}(w_{\pi(k)}-w_{\pi(k+1)}) \log \right | \boldsymbol{I} \\
& \displaystyle ~~~~~~~~~~~~~~~~~~ + \left . \sum_{u=1}^{k}{\boldsymbol{H}_{\pi(u)} \boldsymbol{Q}_{i, \pi(u)} \boldsymbol{H}_{\pi(u)}^{\dagger} } \right | \\
\displaystyle \text{s. t.} & \displaystyle \text{tr}(\boldsymbol{Q}_{i, k})=P_{i, k}, \;\; k=1, \ldots, K, \\
\end{array}
\end{equation}
where $\pi$ is the permutation of user indices $\{1, \ldots, K\}$ such that $w_{\pi(1)} \geq w_{\pi(2)} \geq \cdots \geq w_{\pi(K)}$, and $w_{\pi(K+1)}=0$.}
\end{lemma}
\begin{proof}
See Appendix~A.
\end{proof}

Using $R(\boldsymbol{P}_i)$, we can then reformulate (\ref{eq.p}) into the following power allocation problem:
\begin{equation}\label{eq.p1}
\begin{array}{cll}
\displaystyle \max_{\{ \boldsymbol{P}_i \}} & \displaystyle \sum_{i=1}^{N}{\left [ R(\boldsymbol{P}_i) L_i \right ]} \\
\displaystyle \text{s. t.} & \displaystyle \sum_{i=1}^{n}(P_{i, k} L_i) \leq E_{n, k}^{a}, & \forall n, \;\; \forall k, \\
\displaystyle & \displaystyle \sum_{i=1}^{n}(P_{i, k} L_i) \geq E_{n, k}^{c}, & \forall n, \;\; \forall k, \\
\displaystyle & \displaystyle P_{i,k} \geq 0, & \forall n, \;\; \forall k.
\end{array}
\end{equation}

\subsection{Optimality Conditions}

Since $R(\boldsymbol{P}_i)$ is a concave function of $\boldsymbol{P}_i$ per Lemma~\ref{lemma1}, it then follows that \eqref{eq.p1} is a convex problem. Let $\boldsymbol{\Lambda}:=\{\lambda_{n, k}, \mu_{n, k}, n=1,\ldots, N, k = 1, \ldots, K\}$, where $\lambda_{n, k}$ and $\mu_{n, k}$ denote the Lagrange multipliers associated with the causality and non-overflow constraints of user $k$, respectively. The Lagrangian of (\ref{eq.p1}) is given by:
\begin{equation}\nonumber
\begin{aligned}
& \mathcal{L}(\boldsymbol{P}, \boldsymbol{\Lambda}) = \displaystyle\sum_{i=1}^{N}{[R(\boldsymbol{P}_i) L_i]} \\
& ~~~~~~~~~~~~~~ - \displaystyle\sum_{n=1}^{N}{\sum_{k=1}^{K}{\lambda_{n, k} \left ( \displaystyle\sum_{i=1}^{n}{P_{i, k}L_i-E_{n, k}^{a}} \right )}} \\
& ~~~~~~~~~~~~~~ + \displaystyle\sum_{n=1}^{N}{\displaystyle\sum_{k=1}^{K}{\mu_{n, k} \left (\displaystyle\sum_{i=1}^{n}{P_{i, k}L_i-E_{n, k}^{c}} \right )}} \\
& = \displaystyle\sum_{i=1}^{N} \left [R(\boldsymbol{P}_i) -\displaystyle\sum_{k=1}^{K}{\left (\displaystyle\sum_{n=i}^{N}{\lambda_{n, k}}-\displaystyle\sum_{n=i}^{N}{\displaystyle\mu_{n, k}} \right ) P_{i, k}L_i} \right ] + \mathcal{C}(\boldsymbol{\Lambda})
\end{aligned}
\end{equation}
where $\mathcal{C}(\boldsymbol{\Lambda})=\sum_{n=i}^{N}{(\lambda_{n, k}E_{n, k}^{a})}-\sum_{n=i}^{N}{(\mu_{n, k}E_{n, k}^{c})}$.

Let $\boldsymbol{P}_i^*$ denote the optimal solution for (\ref{eq.p1}) and $\boldsymbol{\Lambda}^*$ denote the optimal Lagrange multiplier. Define $\theta_{i,k} := \sum_{n=i}^N \lambda_{n,k}^* - \sum_{n=i}^{N} \mu_{n,k}^*$. Based on the Karush-Kuhn-Tucker (KKT) optimality conditions \cite{convex}, we must have: $\forall i$, $\forall k$,
\begin{equation}\label{eq.kkt1p}
\begin{aligned}
\boldsymbol{P}_i^* & =\arg \max_{\boldsymbol{P}_i \succeq 0}{\left [R(\boldsymbol{P}_i) L_i-\sum_{k=1}^{K}{\theta_{i, k} P_{i, k} L_i} \right ]} \\
& =\arg \max_{\boldsymbol{P}_i \succeq 0}{\left [R(\boldsymbol{P}_i)-\sum_{k=1}^{K}{\theta_{i, k} P_{i, k}} \right ]}. \\
\end{aligned}
\end{equation}
In addition, the non-negative multipliers $\lambda_{n,k}^*$ and $\mu_{n,k}^*$ satisfy the complementary slackness conditions: $\forall n, \forall k$,
\begin{equation}\label{eq.kkt2p}
\left \{
\begin{array}{c}
\lambda_{n, k}^*=0, \;\; {\text{if}} \;\; \sum_{i=1}^{n}{(P_{i, k}^* L_i) < E_{n, k}^a}; \\
\sum_{i=1}^{n}{(P_{i, k}^* L_i)=E_{n, k}^a}, \;\; {\text{if}} \;\; \lambda_{n, k}^* > 0; \\
\end{array}
\right.
\;\; \forall{n}, \;\; \forall{k}.
\end{equation}
\begin{equation}\label{eq.kkt3p}
\left \{
\begin{array}{c}
\mu_{n, k}^*=0, \;\; {\text{if}} \;\; \sum_{i=1}^{n}{(P_{i, k}^* L_i) > E_{n, k}^c}; \\
\sum_{i=1}^{n}{(P_{i, k}^* L_i)=E_{n, k}^c}, \;\; {\text{if}} \;\; \mu_{n, k}^* > 0; \\
\end{array}
\right.
\;\; \forall{n}, \;\; \forall{k}.
\end{equation}

Since $R(\boldsymbol{P}_i)$ is not given in closed-form, the globally optimal $\{ \boldsymbol{P}_i^* \}$ for \eqref{eq.p1} satisfying \eqref{eq.kkt1p}--\eqref{eq.kkt3p} are challenging to find by general convex program solvers. Relying on the specific structure revealed by \eqref{eq.kkt1p}--\eqref{eq.kkt3p}, we next develop a low complexity block coordinate ascent method to obtain $\{ \boldsymbol{P}_i^* \}$.

\subsection{Optimal User Power Allocation}

Different from the point-to-point and broadcast transmissions, each user in the MAC has its own set of energy causality and non-overflow constraints. These $K$ sets of energy harvesting constraints could have coupled effect on the optimal user transmission strategies. To bypass this difficulty, we resort to a sequential optimization manner. To be specific, we find the optimal power allocation for a single user with all other users' powers being fixed per iteration.

Let $\boldsymbol{P}_{i, -k}$ collect all the power values in $\boldsymbol{P}_i$ except for $P_{i,k}$, and rewrite $R^{(q)}(P_{i,k}, \boldsymbol{P}_{i, -k}):=R(\boldsymbol{P}_i)$. Let  $q$ denote the iteration index, $P_{i,k}^{(q)}$ the optimal value of $P_{i,k}$, and $\lambda_{n, k}^{(q)}$, $\mu_{n, k}^{(q)}$ the corresponding optimal Lagrange multipliers in the $q$th iteration. With $\boldsymbol{P}_{i, -k}$ fixed, it follows from \eqref{eq.kkt1p} that
\begin{equation}\label{eq.kkt1r}
\begin{aligned}
P_{i,k}^{(q)} & =\arg \max_{P_{i,k} \geq 0}{\left [R^{(q)}(P_{i,k}, \boldsymbol{P}_{i, -k})-\sum_{u=1}^{K}{\theta_{i, u}^{(q)} P_{i, u}} \right ]} \\
& =\arg \max_{P_{i,k} \geq 0}{\left [R^{(q)}(P_{i,k}, \boldsymbol{P}_{i, -k}) - \theta_{i, k}^{(q)} P_{i, k} \right ]}. \\
\end{aligned}
\end{equation}
On the other hand, directly from (\ref{eq.kkt2p})--(\ref{eq.kkt3p}), we must have:
\begin{equation}\label{eq.kkt2r}
\left \{
\begin{array}{c}
\lambda_{n, k}^{(q)}=0, \;\; {\text{if}} \;\; \sum_{i=1}^{n}{(P_{i, k}^{(q)} L_i) < E_{n, k}^a}; \\
\sum_{i=1}^{n}{(P_{i, k}^{(q)} L_i) = E_{n, k}^a}, \;\; {\text{if}} \;\; \lambda_{n, k}^{(q)} > 0;
\end{array}
\right.
\;\forall{n};~~~
\end{equation}
\begin{equation}\label{eq.kkt3r}
\left \{
\begin{array}{c}
\mu_{n, k}^{(q)}=0, \;\; {\text{if}} \;\; \sum_{i=1}^{n}{(P_{i, k}^{(q)} L_i) > E_{n, k}^c}; \\
\sum_{i=1}^{n}{(P_{i, k}^{(q)} L_i) = E_{n, k}^c}, \;\; {\text{if}} \;\; \mu_{n, k}^{(q)} > 0;
\end{array}
\right.
\;\forall{n}.~~~
\end{equation}

Interestingly, \eqref{eq.kkt1r}--\eqref{eq.kkt3r} correspond to the KKT conditions with the total-throughput maximization problem for an equivalent ``point-to-point'' link between user $k$ and the access point. Hence, the ``string-tautening'' algorithm in our recent work \cite{Wang14} could be applied to obtain $\{ \boldsymbol{P}_i^* \}$ with a low computational complexity. Note that the signals from other users, if not cancelled, are treated as noises in the optimal MIMO MAC decoding. Therefore, although the physical MIMO MAC is time-invariant here, the resultant equivalent point-to-point link obtained by fixing other user¡¯s powers is in fact ``time-varying'' due to the possibly different noise levels introduced by different signal powers with other users per epoch. As a result, we should employ a water-level based ``string-tautening'' algorithm to find $\{ \boldsymbol{P}_i^* \}$.

Define $g(P_{i,k}):=\frac{\partial R^{(q)}(P_{i,k}, \boldsymbol{P}_{i, -k})}{\partial P_{i, k}}$,
and let $g^{-1}$ denote the inverse function of $g$. It can be inferred from \eqref{eq.kkt1r}--\eqref{eq.kkt3r} that:
\begin{lemma}\label{lemma2}
{\it
In the optimal transmission strategy, we have the following two propers:

1) The optimal power values $\{P_{i,k}^{(q)}\}$ are given by a water-filling alike form: $P_{i,k}^{(q)}=g^{-1(q)}(1/ \omega_{i,k}^{(q)})$, where the ``water-level'' $\omega_{i,k}^{(q)}:=1/ \theta_{i, k}^{(q)}$, $\forall i$.

2) The water-level $\omega_{i,k}^{(q)}$ only changes at some $t_n$ where the causality or non-overflow constraints are tight; specifically, it increases after a $t_n$ where $\sum_{i=1}^{n} (P_{i,k}^{(q)} L_i) = E_{n,k}^a$, and it decreases after a $t_n$ where $\sum_{i=1}^{n} (P_{i,k}^{(q)} L_i) = E_{n,k}^c$.
}
\end{lemma}
\begin{proof}
See Appendix~B.
\end{proof}

Since $R(\boldsymbol{P}_i)$ is strictly concave, $g(P_{i,k})$ is strictly decreasing in $P_{i,k}$. Hence, $P_{i,k}^{(q)}=g^{-1(q)}(\theta_{i, k}^{(q)}) =g^{-1(q)}(1/ \omega_{i,k}^{(q)})$ is an increasing function of the water-level $\omega_{i,k}^{(q)}$. In the optimal user power allocation, it is not difficult to see that we shall maintain a constant water-level whenever it is possible. Note that by the water-filling principle, a constant water-level leads to a higher power for an epoch with higher quality to achieve the most efficient usage of available resources. We have to change the water-level when the energy causality or non-overflow constraints become tight. A causality constraint is tight, i.e., all available energy is used up at $t_n$ when the energy harvested so far is not enough; as a result, a lower water-level (thus power) is maintained before $t_n$ than that after. Similarly, a non-overflow constraint is tight at $t_n$ when the energy harvested so far is abundant, then a higher water-level (thus power) is maintained before $t_n$ than that after.

Based on the structure revealed in Lemma~2, we next develop a water-level based ``string-tautening'' algorithm. Let $\omega_{n,k}^{+ (q)}$ and $\omega_{n,k}^{- (q)}$ denote the constant water-levels to make the $n$th causality and non-overflow constraints of user $k$ become tight at $t_n$ in the $q$th iteration, respectively. Given an invariant water-level $\omega$ before $t_n$, the user power per epoch $i$ is given by $P_{i,k}^{(q)}=g^{-1 (q)}(1/ \omega)$. Thus, the values of $\omega_{n,k}^{+ (q)}$ and $\omega_{n,k}^{- (q)}$, $\forall n$, $\forall k$, can be calculated by solving the equations: $\forall n$, $\forall k$,
\begin{equation}\label{eq.w}
\begin{array}{c}
\displaystyle \sum_{i=1}^{n} \left [\left (g^{-1 (q)} \left (\textstyle \frac{1}{\omega_{i,k}^{+ (q)}} \right ) \right ) L_i \right ] = E_{n,k}^a; \\
\displaystyle \sum_{i=1}^{n} \left [\left (g^{-1 (q)} \left (\textstyle \frac{1}{\omega_{i,k}^{- (q)}} \right ) \right ) L_i \right ] = E_{n,k}^c. \\
\end{array}
\end{equation}
Note that $\sum_{i=1}^{n} [g^{-1 (q)}(1/ \omega) L_i]$ is increasing in $\omega$. Hence, the equations in \eqref{eq.w} can be solved by a bisection search.

Upon with $\omega_{n,k}^{+ (q)}$ and $\omega_{n,k}^{- (q)}$ obtained, we are ready for implementation of the water-level based ``string-tautening'' scheme. Define $\boldsymbol{\mathcal{E}}_k:=\{E_{0,k}, E_{1,k}, \ldots, E_{N,k}\}$ and $\boldsymbol{\mathcal{L}}:=\{L_1, \ldots, L_N\}$. The powers $P_{i,k}^{(q)}$, $i=1, \ldots, N$ can be obtained by calling Procedure {\bf Schedule}($\boldsymbol{\mathcal{E}}_k$, $\boldsymbol{\mathcal{L}}$, $k$) in Algorithm~1.

\begin{algorithm}[h]
\caption{Water-Level based String Tautening}
\begin{algorithmic}[1]

\Procedure {Schedule}{energy set $\boldsymbol{\mathcal{E}}_k$, length set $\boldsymbol{\mathcal{L}}$, user $k$}
	\State $N_{\text{offset}} = 0$, $P_{i,k}^{(q)}=0$, $\forall i$;
	\While {$N_{\text{offset}}<N$}
		\State [$\tau$, $\omega$, $E$] = FirstChangeW($\boldsymbol{\mathcal{E}}_k$, $\boldsymbol{\mathcal{L}}$, $k$);
		\For {$i = 1$ to $\tau$}
			\State $P_{i,k}^{(q)}=g^{-1 (q)}(1/ \omega)$;
		\EndFor
		
		\State $N_{\text{offset}}= N_{\text{offset}}+\tau$;
		\State update $\boldsymbol{\mathcal{E}}_k$, $\boldsymbol{\mathcal{L}}$;
	\EndWhile
\EndProcedure
\Statex

\Function {[$\tau$, $\omega$, $E$] =FirstChangeW}{$\boldsymbol{\mathcal{E}}_k$, $\boldsymbol{\mathcal{L}}$, $k$}
    \State $\omega^-=0$, $\omega^+ = \infty$, $\tau^- = \tau^+ = 0$; $N_e = |\mathcal{E}_k|$;
    \For {$n=1$ to $N_e$}
        \State obtain $\omega_{n,k}^{+ (q)}$ and $\omega_{n,k}^{- (q)}$ by solving equations in \eqref{eq.w};
        \If {$\omega_{n,k}^{+ (q)} \leq \omega^+$}
            \State $\tau^+ = n$, $\omega^+ = \omega_{n,k}^{+ (q)}$, $E^+ = E_{\tau,k}^a$;
        \EndIf
        \If {$\omega_{n,k}^{- (q)} \geq \omega^-$}
            \State $\tau^- = n$, $\omega^- =\omega_{n,k}^{- (q)}$, $E^- = E_{\tau,k}^c$;
        \EndIf

        \If {$\omega^- > \omega^+$ \& $\tau^- < \tau^+$}
            \State return $\tau =\tau^-$, $\omega = \omega^-$, $E = E^-$;
        \ElsIf {($\omega^- \geq \omega^+$ \& $\tau^- \geq \tau^+$) or ($\tau^+ = N_e$)}
            \State return $\tau =\tau^+$, $\omega = \omega^+$, $E = E^+$;
        \EndIf
    \EndFor
\EndFunction
\end{algorithmic}
\end{algorithm}

The key component in Algorithm~1 is the function FirstChange, which determines the first water-level changing time $t_{\tau}$ and the water-level $\omega$ used before $t_{\tau}$. The two candidate water-levels are updated as: $\omega^+=\min_{i \leq n}\omega_{i,k}^{+ (q)}$, and $\omega^-=\max_{i \leq n}\omega_{i,k}^{- (q)}$. If we have $\omega^+ < \omega^-$ at a certain time $t_n$, then the water-level needs to be changed before $t_n$ since no invariant water-level can satisfy all the causality and non-overflow constraints so far. The first water-level changing time can be obtained by comparing $\tau^-$ and $\tau^+$ to see which type of constraint first becomes tight. When the returned $t_{\tau}< T$, the function FirstChange can be reused for a new ($\boldsymbol{\mathcal{E}}_k$, $\boldsymbol{\mathcal{L}}$, $k$) system over the remaining time to find the next water-level.

\begin{proposition}
{\it Algorithm~1 can compute $P_{i,k}^{(q)}$, $i=1, \ldots, N$ with a complexity $\mathcal{O}(N)$.}
\end{proposition}
\begin{proof}
See Appendix~C.
\end{proof}
%
%

\subsection{Block Coordinate Ascent}\label{sec_3_c}

Define $\boldsymbol{\mathcal{E}}:=\{ \boldsymbol{\mathcal{E}}_1, \ldots, \boldsymbol{\mathcal{E}}_K \}$. Building on Algorithm 1, we propose to solve \eqref{eq.p1} by calling Procedure {\bf Iteration}($\boldsymbol{\mathcal{E}}$, $\boldsymbol{\mathcal{L}}$) in Algorithm~2.

\begin{algorithm}[h]
\caption{Block Coordinate Ascent}
\begin{algorithmic}[1]

\Procedure {Iteration}{energy set $\boldsymbol{\mathcal{E}}$, length set $\boldsymbol{\mathcal{L}}$}
	\State select a tolerance level $\epsilon > 0$, and $q=1$;
	\State find a set of feasible $\{ P_{i,k}^{(0)} \}$, $\forall i$, $\forall k$;
	\State calculate initial throughput $W^{(0)}$;
	
	\While {$1 > 0$}
		\For {$k = 1$ to $K$}
		\State calculate $P_{i,k}^{(q)}$, $\forall i$ by calling Procedure {\bf Schedule}($\boldsymbol{\mathcal{E}}_k$, $\boldsymbol{\mathcal{L}}$, $k$);
		\EndFor
		
		\State calculate throughput $W^{(q)}$;
		
		\If {$ (W^{(q)}-W^{(q-1)})/W^{(q-1)} \leq \epsilon$}
			\State $\boldsymbol{P}_i^*=[P_{i,1}^{(q)}, \ldots, P_{i, K}^{(q)}]$, $\forall i$;
			\State {\bf break};
		\EndIf
		
		\State $q = q+1$;
	\EndWhile
\EndProcedure
\end{algorithmic}
\end{algorithm}

Per iteration $q$ in Algorithm 2, we sequentially optimize the users' transmit-powers one-by-one using the water-level based string-tautening procedure in Algorithm~1. Based on $\boldsymbol{P}_i^{(q)} = [P_{i,1}^{(q)}, P_{i,2}^{(q)}, \ldots, P_{i,K}^{(q)}]$, we solve \eqref{eq.Rp} to compute $R(\boldsymbol{P}_i^{(q)})$ per Lemma 1. Then we calculate the sum-throughput $W^{(q)}=\sum_{i=1}^N [R(\boldsymbol{P}_i^{(q)})L_i]$ and compare it with $W^{(q-1)}$ from last iteration. The iteration is terminated when the increment of sum-throughput is less than the tolerance level. Since Algorithm~2 in fact follows the classic block coordinate ascent method, it converges to at least a local optimum. As (\ref{eq.p1}) is a convex problem, every local optimum is global optimum; i.e., Algorithm~2 converges to the globally optimal $\boldsymbol{P}^*$ for (\ref{eq.p1}). Having obtained $\boldsymbol{P}^*$, we can consequently find the optimal transmit-covariance matrices $\{\boldsymbol{Q}_{i,k}^*\}$ by solving (\ref{eq.Rp}) with $P_{i,k} \equiv P_{i,k}^*$, $\forall i,k$ for the MIMO MAC. Summarizing, we have:
\begin{theorem}
{\it The proposed block coordinate ascent approach is guaranteed to find the globally optimal MIMO MAC transmission policy for (\ref{eq.p}).}
\end{theorem}

The proposed Algorithm 2 is initialized by a set of feasible $\{ P_{i,k}^{(0)} \}$. Note that Algorithm 2 is guaranteed to converge to $\boldsymbol{P}^*$ with any feasible $P_{i,k}^{(0)}$ per Theorem 1. However, the convergence speed may vary when selecting different initial points. We envision that a good initialization can be obtained as follows. Treat the MIMO MAC as $K$ decoupled ``point-to-point'' links between each user and the access point. Then calculate the optimal power allocation $\{ \tilde{P}_{i,k}^*, \forall i\}$ per link $k$. Note that for the decoupled time-invariant link, we just need a simpler power-based (instead of water-level based) ``string-tautening'' algorithm to find $\{ \tilde{P}_{i,k}^*, \forall i\}$; this can be vividly illustrated by the trajectory of letting a string tie its one end at the origin $(0,0)$, pass the other end through $(T, \sum_{i=0}^{N-1}{E_{i,k}})$, and then tauten between user $k$'s own energy arrival curve and minimum departure curve \cite{Wang14, Zafer&Modiano2005}. Since each set of $\{ \tilde{P}_{i,k}^*, \forall i\}$ satisfies the energy causality and non-overflow constraints with user $k$, it is clear that the resultant $\boldsymbol{\tilde{P}}_i^*:=[\tilde{P}_{i,1}^*,\ldots,\tilde{P}_{i,K}^*]$ is feasible for (\ref{eq.p1}). Interestingly, simulations will show the sum-throughput achieved by $\{ \boldsymbol{\tilde{P}}_i^* \}$ is always very close to that with the optimal $\{ \boldsymbol{P}_i^* \}$; hence, it provides a good initial point.

\subsection{Visualization of the Optimal Policy}

The optimal policy obtained by Algorithm~2 can be visualized by using a calculus approach similar to the ones in \cite{Zafer&Modiano2005, Wang12, Nan16}, which was developed for energy-efficient transmissions of bursty data arrivals with delay constraints. Mimicking the definitions of data arrival and minimum departure curves in \cite{Zafer&Modiano2005, Wang12, Nan16}, we define the energy arrival curve $A_k(t)$ and minimum energy departure curve $D_{\text{min}, k}(t)$ of user $k$ as: $\forall k$,
\begin{equation}
\begin{array}{cl}
\displaystyle A_k(t) = \sum_{i = 0}^{N - 1} \left [ E_{i, k} u(t - t_i) \right ], & 0 \leq t \leq T, \\
\displaystyle D_{\text{min}, k}(t) = \left ( \sum_{i = 0}^{N - 1} \left [ E_{i, k} u(t - t_i) \right ] \right )^+, & 0 \leq t \leq T,
\end{array}
\end{equation}
where $u(t)$ is the unit-step function: $u(t) = 1$ if $t \geq 0$, and $u(t) = 0$ otherwise. Consider a piece-wise linear energy departure curve: $\forall k$,
\begin{equation}\label{eq.D_k}
D_k(t) = \sum_{i = 1}^{n - 1} P_{i, k} L_i + P_{n, k} (t - t_{n - 1}), \quad \forall n \in [1, N],
\end{equation}
where the user powers $P_{i, k}$ per epoch serve as the piece-wise non-negative slopes for $D_k(t)$. As mentioned previously, we can obtain the initial energy departure curve $\tilde{D}_k(t)$ of user $k$ by the trajectory of letting a string tie its one end at the origin $(0, 0)$, pass the other end through $(T, \sum_{i = 0}^{N - 1} E_{i, k})$, and then tauten between $A_k(t)$ and $D_{\text{min}, k}(t)$; see the two-user example in Fig.~\ref{fig.energy_departure}. Consequently, the initial user powers $\{ \tilde{P}_{i,k}^* \}$ (or $\{ P_{i,k}^{(0)} \}$) can be derived from $D^{(0)}_k(t)$, $\forall k$.

Substituting the optimal user powers $\{ P^*_{i, k} \}$, $\forall i$, $\forall k$ obtained by Algorithm~2 into \eqref{eq.D_k}, we produce the optimal energy departure curves $D^*_k(t)$ of the $K$ users; see Fig.~\ref{fig.energy_departure}. The differences between the two policies are obviously seen. Recall that $\tilde{D}_k(t)$ is also regarded as the optimal transmission policy for the decoupled point-to-point link between the $k$th user and the access point. In this policy, the transmit power $\tilde{P}^*_{i, k}$ of user $k$ only changes at its {\it non-zero} energy arrival instants \cite{SECON14}. For example, as can be seen in Fig.~\ref{fig.energy_departure}, $\tilde{P}^*_{i, 2}$ only changes at $t_4 = 8$, $t_5 = 9$ and $t_6 = 4$, where the corresponding amounts of arriving energy $E_{i, 2} > 0$. As for $D^*_k(t)$, however, since it concerns the coupling between multiple energy harvesting processes, the transmit power $P^*_{i, k}$ of a user can change at both the non-zero energy arrival instants of its own and those of all other users.

\begin{figure}[h]
\centering
\includegraphics[width=3.6in]{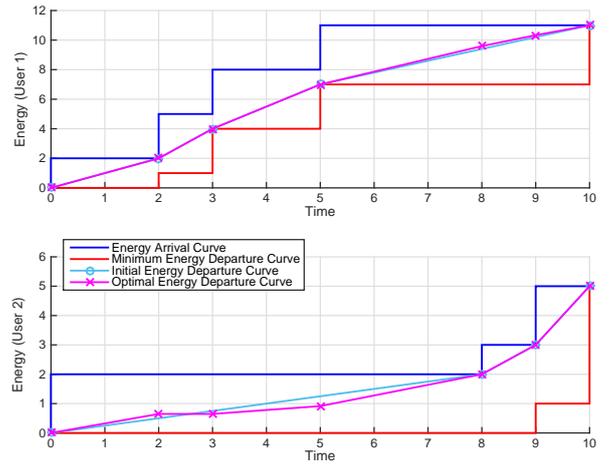}
\caption{Energy arrival, minimum energy departure, initial energy departure and optimal energy departure curves of the users.}\label{fig.energy_departure}
\end{figure}

\subsection{Heuristic Online Scheme}

The near-optimality of the decoupled ``power-tautening'' method motivates us to develop a heuristic online scheme. Suppose that the time-invariant $\mathcal{H}$ is known, and the harvested energy of each user is modeled by a compound Poisson process, where the number of energy arrivals over a horizon $T$ follows a Poisson distribution with mean $\lambda_e$ and the energy amount in each arrival is independent and identically distributed (i.i.d.) with mean $\bar{E}_k$ \cite{Xu13}. It is assumed that $\lambda_e$ and $\bar{E}_k$ are known a priori or that they can be estimated on-the-fly by the average energy arrival interval and amount from previous energy harvesting process.

While the initial amounts of energy $E_{0,k}$, $\forall k$ are known, all $E_{i,k}$, $i=1, \ldots, n-1$, $\forall k$, and $L_i$, $i=1, \ldots, n$, are clearly non-causal information that is seldom available a-priori in practice. However, it is worth noting that $\sum_{i=1}^{n-1}E_{i,k}$ is the total energy harvested by user $k$ during $(0,t_n)$, where $t_n = \sum_{i=1}^n L_i$. Then, given $\lambda_e$ and $\bar{E}_k$, we have:
\begin{equation}\label{eq.approx}
\sum_{i=1}^n L_i \approx \frac{n}{\lambda_e}
\end{equation}
and
\begin{equation}\label{eq.approx2}
\frac{\sum_{i=1}^{n-1}E_{i,k}}{\sum_{i=1}^n L_i} \approx \frac{(n-1) \bar{E}_k}{n/\lambda_e } = \frac{n-1}{n} \lambda_e \bar{E}_k,
\end{equation}
where the approximation becomes exact as $t_i \rightarrow \infty$. Let $P_{n,k}^+$ and $P_{n,k}^-$ denote the constant power to make the $n$th causality and non-overflow constraints of user $k$ become tight at $t_n$, respectively. Using (\ref{eq.approx})--(\ref{eq.approx2}), we can then approximate:
\begin{equation}\label{eq.app1}
P_{n,k}^+ \approx \lambda_e \frac{E_{0,k} + (n-1) \bar{E}_k}{n},
\end{equation}
\begin{equation}\label{eq.app2}
P_{n,k}^- \approx \lambda_e \frac{ (E_{0,k} + n \bar{E}_k-E_{\max,k})^+}{n}.
\end{equation}
Using the approximations in (\ref{eq.app1}) and (\ref{eq.app2}), we can implement the power-tautening method in the same spirit with Algorithm~1 to produce the first set of powers adopted by the users at $t_0$. This set of powers is kept in use before new energy arrives at $t_1$. Then we treat $t_1$ as new ``$t_0$'', and update the initial $E_{0,k}$ with the sum of unconsumed energy (if any) and newly arriving energy amount $E_{1,k}$, $\forall k$. Relying on (\ref{eq.app1})--(\ref{eq.app2}), the algorithm is run again to yield the next set of user powers. This process continues until all energy is used up or the end of transmission period $T$ is reached.

\section{Numerical Results}\label{sec4}

\begin{figure}[h]
\centering
\includegraphics[width=3.6in]{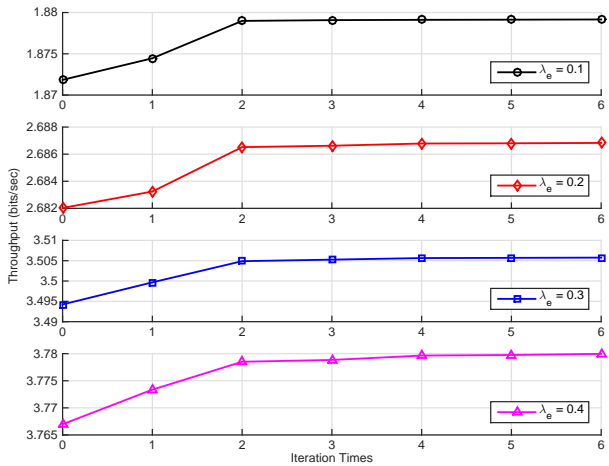}
\caption{Average throughput after each time of iteration.}\label{convergence}
\end{figure}

For a time-invariant MIMO MAC with $K=2$ users, consider data transmission over $T=20$ seconds. The weight vector $\boldsymbol{w}=[1, 1]$, and each element in channel matrix $\boldsymbol{H}_{k}$, $k=1,2$, is a zero-mean complex Gaussian random variable with unit variance. The battery capacity of each user is $E_{\text{max,1}}=E_{\text{max,2}}=10$ Joules. Assume each user's energy harvesting setup is modeled by a compound Poisson process with the same mean $\lambda_e$. The amount of energy in each arrival is assumed to be independent and uniformly distributed with mean $5$ Joules. 
Fig.~\ref{convergence} shows the average throughput after each iteration of the proposed Algorithm 2 when $(N_t,N_r)$ is set to $(2,2)$, and $\lambda_e=0.1$, $0.2$, $0.3$, or $0.4$ $\text{sec}^{-1}$. 
The fast convergence of Algorithm 2 is clearly observed: for each $\lambda_e$ value, the proposed algorithm could always converge to the optimum in 5 iterations.

Fig.~\ref{throughput_1} and Fig.~\ref{throughput_2} show the average throughputs versus $T$ for two different values of $\lambda_e = 0.1$ and $0.3$ $\text{sec}^{-1}$, respectively, where each result is obtained as the average of 40 randomly generated trial cases. We compare the performance of the proposed Algorithm 2 with that of the decoupled power-tautening (P-Tautening) scheme, and the proposed online scheme. In addition, we also include the performance of two other feasible schemes for comparison: Causality-Satisfied and Non-Overflow, which are obtained by always selecting each user's power to meet its next causality and non-overflow constraints, respectively. The sum-throughputs of all the five schemes slightly increase as the transmit duration $T$ extends from $10$ to $50$ sec. With full energy-harvesting information available a-prior, Algorithm 2 always provides the optimal benchmark for all values of $T$ in both cases. It is interestingly observed that the decoupled P-Tautening scheme, in which each user only requires its own energy-harvesting information, could achieve a throughput that is over $98\%$ of the optimal benchmark, as accurately revealed by Table~\ref{performance}. This observation just corroborates the merit of our initial point selecting scheme in Section~\ref{sec_3_c}. Requiring only the next (non-causal) energy arrival information, Causality-Satisfied and Non-Overflow schemes both incur roughly over $0.3$ bits/sec throughput loss. The proposed online scheme could always attain a fairly good $85\%$ of the optimal sum-throughput for all $T$ values, even though only causal information is required by each user.

\begin{figure}[t!]
\centering
\includegraphics[width=3.6in]{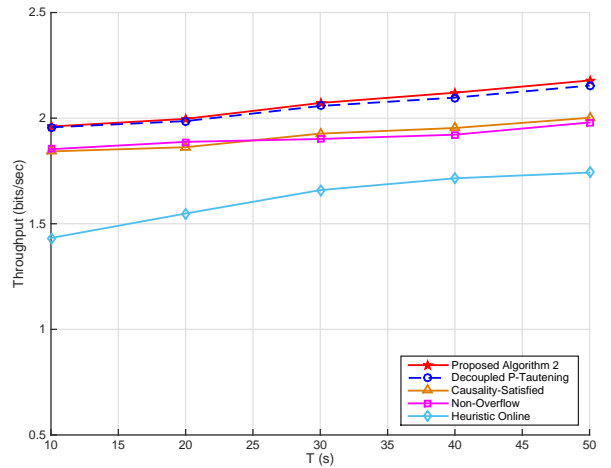}
\caption{Average throughput vs. transmit duration $T$ ($\lambda_e = 0.1$, $N_r = 2$).}\label{throughput_1}
\end{figure}
\begin{figure}[t!]
\centering
\includegraphics[width=3.6in]{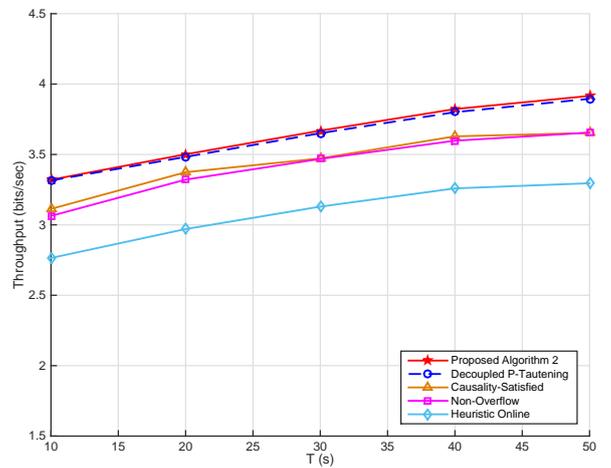}
\caption{Average throughput vs. transmit duration $T$ ($\lambda_e = 0.3$, $N_r = 2$).}\label{throughput_2}
\end{figure}

For further illustration, we next depict the influence of the energy arrival rate on the sum-throughputs in Fig.~\ref{throughput_3} and Fig.~\ref{throughput_4}. We test two scenarios: $(N_t, N_r) = (2, 2)$ and $(N_t, N_r) = (2, 4)$ for all the five transmit schemes with a fixed transmit duration $T = 20$ sec. It is clearly seen that the sum-throughputs of all the five schemes increase as $\lambda_e$ increases, since there can be more energy available at the users when $\lambda_e$ becomes larger. For small $\lambda_e$ values, the performance of the decoupled P-Tautening scheme almost reaches the optimal benchmark, while as $\lambda_e$ increases, without concerning the coupling effect of users, it gradually becomes sub-optimal. It is also observed that the sum-throughput is significantly improved for the MIMO MAC as the number of receive antennas $N_r$ doubles.

\begin{figure}[t!]
\centering
\includegraphics[width=3.6in]{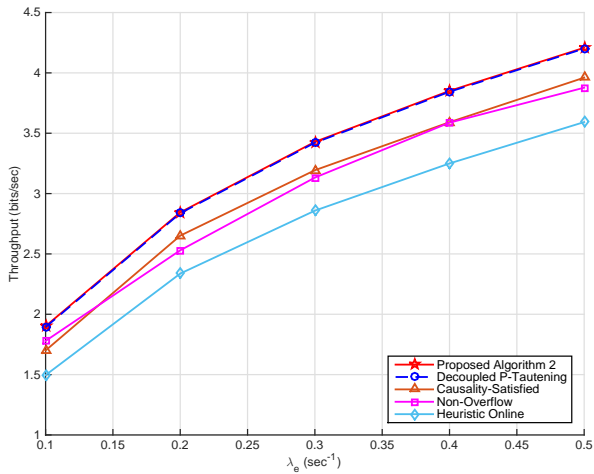}
\caption{Average throughput vs. energy arrival rate $\lambda_e$ ($T = 20$, $N_r = 2$).}\label{throughput_3}
\end{figure}
\begin{figure}[t!]
\centering
\includegraphics[width=3.6in]{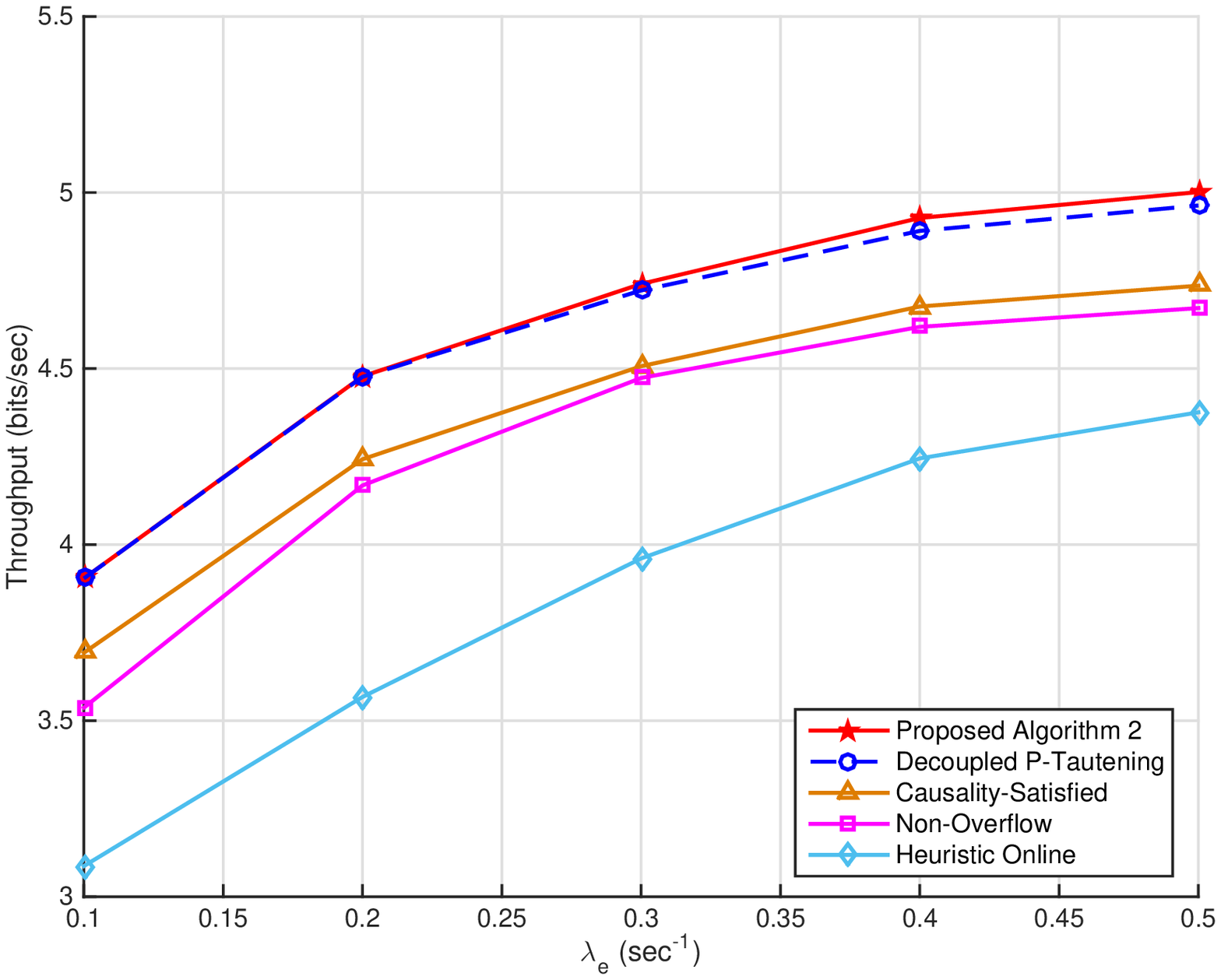}
\caption{Average throughput vs. energy arrival rate $\lambda_e$ ($T = 20$, $N_r = 4$).}\label{throughput_4}
\end{figure}
\begin{table}[b]
\caption{Average throughputs of Algorithm~2 and Decoupled P-Tautening vs. $T$.}\label{performance}
\centering
\begin{tabular}{r | r r | r r}
\toprule
\multirow{2}{*}{$T$ (sec)} & Algorithm~2 & P-Tautening & Algorithm~2 & P-Tautening \\
& \multicolumn{2}{c |}{$\lambda_e = 0.1$ $\text{sec}^{-1}$} & \multicolumn{2}{c}{$\lambda_e = 0.3$ $\text{sec}^{-1}$} \\
\midrule
10 & 1.9603 & 1.9564 & 3.3204 & 3.3150 \\
20 & 1.9967 & 1.9859 & 3.5012 & 3.4828 \\
30 & 2.0719 & 2.0576 & 3.6691 & 3.6510 \\
40 & 2.1205 & 2.0972 & 3.8225 & 3.8007 \\
50 & 2.1784 & 2.1541 & 3.9159 & 3.8953 \\
\botrule
\end{tabular}
\end{table}

\section{Concluding Remarks}\label{sec5}

We proposed a novel approach to obtaining the optimal transmission policy for energy harvesting powered MIMO MACs. An efficient algorithm was developed to find the optimal offline solution with a low computational complexity. Our approach can provide the optimal benchmarks for practical schemes. Development of online scheme based on the revealed optimal structure was also discussed.

In our proposed approach, the energy harvesting processes were modeled as deterministic processes, where the amount of each energy arrival is accurately known a-priori. Due to the unpredictable and intermittent nature of the renewable energy sources, the amount of harvested energy may not always be a-priori available in practice \cite{Hu15, Hu16}. Taking this uncertainty into account, we may obtain a result of improved practical significance. This will be an interesting direction to pursue in our future work.

As an alternative promising solution for achieving energy efficiency, the smart-grid industry has seen fast growth in the past decade \cite{ChenMag17}. It is expected that future wireless communication systems will be equipped with energy harvesting devices and powered by smart-grids. Energy management strategies of smart-grid powered coordinated multi-point (CoMP) systems and MIMO broadcasting systems have been explored in \cite{JSAC16_1, JSAC16_2, Wang17}. Building on the proposed approach, system design and power control of smart-grid involved MIMO MACs are also worth investigating.

\section{Appendices}\label{sec14}

\subsection{Proof of Lemma~1}

To show the strict concavity of $R(\boldsymbol{P}_i)$, consider the Lagrangian of (\ref{eq.Rp}):
\begin{equation}
L(\boldsymbol{Q}_i,\boldsymbol{\lambda})=f(\boldsymbol{Q}_i)-\sum_{k=1}^K\lambda_k \left (\text{tr}(\boldsymbol{Q}_{i,k})-P_{i,k} \right ),
\end{equation}
where $f(\boldsymbol{Q}_i)$ denotes the objective function of (\ref{eq.Rp}). Since (\ref{eq.Rp}) is convex program, we have:
\begin{equation}\nonumber
R(\boldsymbol{P}_i)=\min_{\boldsymbol{\lambda}}\max_{\boldsymbol{Q}_i}\left [f(\boldsymbol{Q}_i)-\sum_{k=1}^K\lambda_k \left (\text{tr}(\boldsymbol{Q}_{i,k})-P_{i,k} \right ) \right ].
\end{equation}

For a given $\boldsymbol{P}_i$, let $\boldsymbol{Q}_i^*(\boldsymbol{P}_i)$ and $\boldsymbol{\lambda}^*(\boldsymbol{P}_i)$  denote the optimal primal and dual variables for (\ref{eq.Rp}). For $\tilde{\boldsymbol{P}}_i,\hat{\boldsymbol{P}}_i$, and $\check{\boldsymbol{P}}_i=\beta\tilde{\boldsymbol{P}}_i+(1-\beta)\hat{\boldsymbol{P}}_i$ with $\beta \in [0,1]$, we have:
\begin{equation}\label{eq.apd1}
\begin{aligned}
R(\tilde{\boldsymbol{P}}_i) = & f(\boldsymbol{Q}_i^*(\tilde{\boldsymbol{P}}_i))-\sum_{k}\lambda_k^*(\tilde{\boldsymbol{P}}_i)\left [\text{tr}(\boldsymbol{Q}_{i,k}^*(\tilde{\boldsymbol{P}}_i))-\tilde{P}_{i,k} \right ] \\
\leq & {f(\boldsymbol{Q}_i^*(\tilde{\boldsymbol{P}}_i))-\sum_{k}\lambda_k^*(\check{\boldsymbol{P}}_i)\left [\text{tr}(\boldsymbol{Q}_{i,k}^*(\tilde{\boldsymbol{P}}_i))-\tilde{P}_{i,k} \right ]} \\
\leq & {f(\boldsymbol{Q}_i^*(\check{\boldsymbol{P}}_i))-\sum_{k}\lambda_k^*(\check{\boldsymbol{P}}_i)\left [\text{tr}(\boldsymbol{Q}_{i,k}^*(\check{\boldsymbol{P}}_i))-\tilde{P}_{i,k} \right ]}.
\end{aligned}
\end{equation}
As for \eqref{eq.apd1}, the first inequality follows from that $\boldsymbol{\lambda}^*(\tilde{\boldsymbol{P}}_i)=\min_{\lambda}[f(\boldsymbol{Q}_i^*(\tilde{\boldsymbol{P}}_i))-\sum_{k}\lambda_k(\text{tr}(\boldsymbol{Q}_{i,k}^*(\tilde{\boldsymbol{P}}_i))-\tilde{P}_{i,k})]$,
and the second inequality follows from that $\boldsymbol{Q}_i(\check{\boldsymbol{P}}_i)=\max_{\boldsymbol{Q}}[f(\boldsymbol{Q}_i)-\sum_{k}\lambda_k^*(\check{\boldsymbol{P}}_i)\text{tr}(\boldsymbol{Q}_{i,k})]$.
Similarly, we have:
\begin{align}
   \nonumber
   R(\hat{\boldsymbol{P}}_i)
\leq{f(\boldsymbol{Q}_i^*(\check{\boldsymbol{P}}_i))-\displaystyle\sum_{k}\lambda_k^*(\check{\boldsymbol{P}}_i) \left [\text{tr}(\boldsymbol{Q}_{i,k}^*(\check{\boldsymbol{P}}_i))-\hat{P}_{i,k} \right ]}.
\end{align}
It then follows that:
\[
\begin{aligned}
& \beta{R(\tilde{\boldsymbol{P}}_i)}+(1-\beta)R(\hat{\boldsymbol{P}}_i) \\ \leq & f(\boldsymbol{Q}_i^*(\check{\boldsymbol{P}}_i))-\displaystyle\sum_{k}\lambda_k^*(\check{\boldsymbol{P}}_i) \left [\text{tr}(\boldsymbol{Q}_i^*(\check{\boldsymbol{P}}_i)) \right . \\
& ~~~~~~~~~~~~~~~~~~~~~~~~~~~~~\left . -\beta\tilde{P}_{i,k}+(1-\beta)\hat{P}_{i,k} \right ] \\
= & {f(\boldsymbol{Q}_i^*(\check{\boldsymbol{P}}_i))-\displaystyle\sum_{k}\lambda_k^*(\check{\boldsymbol{P}}_i)\left [\text{tr}(\boldsymbol{Q}_{i,k}^*(\check{\boldsymbol{P}}_i))-\check{P}_{i,k} \right ]} \\
= & R(\check{\boldsymbol{P}}_i).
\end{aligned}
\]
In fact, it can be shown that the strict inequality holds if $\beta\in{(0,1)}$ due to the strict concavity of log function in $f(\boldsymbol{Q}_i)$.

\subsection{Proof of Lemma~2}

Due to the strict concavity and increasing of $R(\boldsymbol{P}_i)$, it clearly follows that we should have $g^{(q)}(P_{i,k})=\theta_{i, k}^{(q)}$, leading to: $P_{i,k}^{(q)}=g^{-1(q)}(\theta_{i, k}^{(q)})$.

Clearly $P_{i,k}^{(q)}$ changes only when water-level $\omega_{i,k}^{(q)}$ changes its value. By the definition $\theta_{i,k}^{(q)} := \sum_{n=i}^N \lambda_{n,k}^{(q)}- \sum_{n=i}^{N-1} \mu_{n,k}^{(q)}$, we have: $1/\omega_{1,k}^{(q)}  = 1/(\sum_{n=1}^N \lambda_{n,k}^{(q)} - \sum_{n=1}^{N-1} \mu_{n,k}^{(q)})$; and $\omega_{N,k}^{(q)} = 1/\lambda_{N,k}^{(q)}$. If $\lambda_{n,k}^{(q)} = \mu_{n,k}^{(q)}=0$, $\forall n = 1, \ldots, N-1$, then a constant $\omega_{i,k}^{(q)} = 1/\lambda_{N,k}^{(q)}$ is maintained over all epoches. We have a change of $\omega_{i,k}^{(q)}$ at $t_n$ only when $\lambda_{n,k}^{(q)}>0$ or $\mu_{n,k}^{(q)}>0$ for a certain $n$.

From the complementary slackness conditions (\ref{eq.kkt2p})--(\ref{eq.kkt3p}), $\lambda_{n,k}^{(q)}>0$ implies that $\sum_{i=1}^{n} (P_{i,k}^{(q)} L_i) = E_{n,k}^a$, whereas $\mu_{n,k}^{(q)}>0$ implies that $\sum_{i=1}^{n}(P_{i,k}^{(q)} L_i) = E_{n,k}^c$. Note that we cannot have both $\lambda_{n,k}^{(q)}>0$ and $\mu_{n,k}^{(q)}>0$ for the same $n$. If $\lambda_{n,k}^{(q)}>0$, we have $1/\omega_{n+1,k}^{(q)} - 1/\omega_{n,k}^{(q)} = -\lambda_{n,k}^{(q)}<0$. This implies that $\omega_{n+1,k}^{(q)} >\omega_{n,k}^{(q)}$. Similarly, if $\mu_{n,k}^{(q)}>0$, we have $1/\omega_{n+1,k}^{(q)} - 1/\omega_{n,k}^{(q)} =\mu_{n,k}^{(q)}>0$, implying $\omega_{n+1,k}^{(q)} <\omega_{n,k}^{(q)}$.

\subsection{Proof of Proposition~1}

Due to the rules used in the function FirstChangeW, we can prove that the water-level changing pattern in the transmission policy produced by Algorithm~1 is consistent with the structure revealed in Lemma~\ref{lemma2}, i.e., 1) if the water-level in use by user $k$ is first $\omega_k$ and then changed to $\tilde{\omega}_k$ at $t_\tau$ where $\sum_{i=1}^{\tau} (g^{-1 (q)}(1/\omega_k) L_i) = E_{n,k}^a$, then we must have $\tilde{\omega}_k > \omega_k$; and 2) if the water-level is changed at $t_{\tau}$ where $\sum_{i=1}^{\tau} (g^{-1 (q)}(1/\omega_k L_i) = E_{n,k}^c$, then we must have the next water-level $\tilde{\omega}_k < \omega_k$.

Suppose w.l.o.g that the water-level of user $k$ changes $J$ times in $\mathcal{W}_k^{(q)} := \{ \omega_{i, k}^{(q)}, i = 1, \ldots, N \}$ yielded by Algorithm~1. We divide the policy into $J + 1$ phases: water-level $\omega_{i, k}^{(q)} = \check{\omega}_{1, k}$ over epoches $i \in [1, \tau_1]$, $\omega_{i, k}^{(q)} = \check{\omega}_{2, k}$ over epoches $i \in [\tau_1 + 1, \tau_2]$, \ldots, $\omega_{i, k}^{(q)} = \check{\omega}_{J + 1, k}$ over epoches $i \in [\tau_J + 1, N]$. We can then construct a set of Lagrange multipliers $\boldsymbol{\Lambda}^{(q)}_k := \{ \lambda_{n, k}^{(q)}, \mu_{n, k}^{(q)}, n = 1, \ldots, N \}$ as follows:

Let $\lambda_{N, k}^{(q)} = g^{(q)}(\check{P}_{J + 1, k}) = 1/\check{\omega}_{J + 1, k} > 0$, where the inequality is due to the strict increasing of $R(P_{i, k}, \boldsymbol{P}_{i, -k})$, leading to positivity of $g(P_{i, k})$. Let $\lambda_{\tau_j, k}^{(q)} = 1/\check{\omega}_{j, k} - 1/\check{\omega}_{j + 1, k}$, if $\sum_{i=1}^{\tau_j} (g^{-1 (q)}(1/\omega_{i, k}^{(q)}) L_i) = E_{\tau_j,k}^a$, or let $\mu_{\tau_j, k}^{(q)} = 1/\check{\omega}_{j + 1, k} - 1/\check{\omega}_{j, k}$ if $\sum_{i=1}^{\tau_j} (g^{-1 (q)}(1/\omega_{i, k}^{(q)}) L_i) = E_{\tau_j,k}^c$, $\forall j = 1, \ldots, J$. We have shown that the water-level $\check{\omega}_{j + 1, k} > \check{\omega}_{j, k}$ if the causality constraint is tight at $t_{\tau_j}$, or $\check{\omega}_{j + 1, k} < \check{\omega}_{j, k}$ if the non-overflow constraint is tight at $t_{\tau_j}$. It then readily follows that $\lambda_{\tau_j, k}^{(q)} > 0$ or $\mu_{\tau_j, k}^{(q)} > 0$, depending which type of constraint is tight at $t_{\tau_j}$. Except for these $J + 1$ positive $\lambda_{N, k}^{(q)}$ and $\lambda_{j, k}^{(q)}$ or $\mu_{j, k}^{(q)}$, all other Lagrange multipliers in $\boldsymbol{\Lambda}^{(q)}_k$ are set to zero.

With such a $\boldsymbol{\Lambda}^{(q)}_k$, the complementary slackness conditions \eqref{eq.kkt2r}--\eqref{eq.kkt3r} clearly hold. Using such a $\boldsymbol{\Lambda}^{(q)}_k$ also leads to $\theta_{i, k}^{(q)} = \sum_{n = i}^N \lambda_{n, k}^{(q)} - \sum_{n = i}^{N} \mu_{n, k}^{(q)} = 1/\check{\omega}_{j, k}$, $\forall i \in [\tau_{j - 1} + 1, \tau_{j}]$ (with $\tau_0 := 1$ and $\tau_{J + 1} := N$). This implies that $P_{i,k}^{(q)} = g^{-1 (q)}(1/\check{\omega}_{j, k}) = \arg \max_{P_{i,k} \geq 0}{[R^{(q)}(P_{i,k}, \boldsymbol{P}_{i, -k}) - \theta_{i, k}^{(q)} P_{i, k}]}$, $\forall i \in [\tau_{j - 1} + 1, \tau_{j}]$. Thus, $\{ P_{i, k}^{(q)}, i = 1, \ldots, N \}$ follows the structure in Lemma~2 and it is guaranteed that each $P_{i, k}^{(q)}$ satisfies \eqref{eq.kkt1r}.

We have proven that $P_{i, k}^{(q)}, i = 1, \ldots, N$ yielded by Algorithm~1 and the Lagrange multipliers $\boldsymbol{\Lambda}^{(q)}_k$ constructed accordingly, satisfy the sufficient and necessary optimality conditions \eqref{eq.kkt1r}--\eqref{eq.kkt3r}. It readily follows that $\mathcal{P}_k^{(q)} := \{ P_{i, k}^{(q)}, i = 1, \ldots, N \}$ is a global optimal policy for the equivalent point-to-point link between user $k$ and the access point corresponding to iteration $q$.

In the search of a water-level changing point and the associated water-level (thus power) per iteration $q$, at most $N$ energy arrival times need to be tested. Clearly, we need to search at most $N$ water-level points in the worst case; hence, the proposed Algorithm~1 can compute $P_{i, k}^{(q)}, i = 1, \ldots, N$ with a complexity $\mathcal{O}(N)$.

\end{document}